\documentclass{article}
\usepackage{latexsym}
\usepackage{colordvi}
\usepackage{amssymb}
\usepackage{graphics}
\usepackage{epsfig}
\def\nn{\mathbb{N}}

\def\NN{\mathbb{N}}

\def\kk{\mathbb{K}}
\def\qq{\mathbb{Q}}
\def\kx{\kk[\mathbf{x}]}%x_1\dots x_n]}
\def\Ac{\mathcal{A}}
\def\Jc{\mathcal{J}}

\def\mc{\mathcal{M}}
\def\Hc{\mathcal{H}}
\def\Rc{\mathcal{R}}
\def\Mon{\mathcal{M}}
\def\x{\mathbf{x}}

\def\tb{\mathbf{t}}
\def\lb{\lambda}
\def\fb{\mathbf{f}}
\def\Lb{\Lambda}

\def\supp{\mathrm{supp}}

\def\lcm{\mathrm{lcm}}
\def\<{\langle}
\def\>{\rangle}
\def\Prod{\Pi}
\def\Prj{\pi}
\def\Prje{\pi^{e}}
\def\deg{\mathrm{deg}}
\def\rewriting{{rewriting~}}
\def\negskip{}
\newcommand{\eqref}[1]{(\ref{#1})}
\newtheorem{lemme}{Lemma}[section]
\newtheorem{defi}[lemme]{Definition}
\newtheorem{assume}[lemme]{Assumption}

\newtheorem{corol}[lemme]{Corrolary}
\newtheorem{theorem}[lemme]{Theorem}
\newtheorem{example}[lemme]{Example}

\newtheorem{remark}[lemme]{Remark}
\newenvironment{proof}{{\bf Proof}.}{\hfill $\Box$}

\def\grob{Gr\"obner }

\title{Stable Normal Forms for Polynomial System Solving}
\author{\begin{tabular}{ll}
Bernard Mourrain & Philippe Tr\'ebuchet \\
\begin{minipage}{7cm}
GALAAD, INRIA M\'editerran\'ee\\
{2004 Route des Lucioles, BP 93,}\\
{06902 Sophia Antipolis, Cedex}, {France}\\
\texttt{Bernard.Mourrain@inria.fr}
\end{minipage}
&
\begin{minipage}{5cm}
{UPMC,LIP6, Equipe APR}\\
{4 place jussieu}\\
{75015 Paris}, {France}\\
\texttt{Philippe.Trebuchet@lip6.fr}
\end{minipage}
 \end{tabular}
}
\begin{document}
\maketitle

\begin{abstract}
  This paper describes and analyzes a method for computing border bases of a
  zero-dimensional ideal $I$. The criterion used in the computation involves
  specific commutation polynomials and leads to an algorithm and an
  implementation extending the one provided in \cite{mt:gnfpss-issac-05}.
  This general border basis algorithm weakens the monomial ordering
  requirement for \grob bases computations. It is up to date the most general
  setting for representing quotient algebras, embedding into a single
  formalism Gr\"obner bases, Macaulay bases and new
  representation that do not fit into the previous categories. With this
  formalism we show how the syzygies of the border basis are generated
  by commutation relations. We also show that our construction of normal form
  is stable under small perturbations of the ideal, if the number of
  solutions remains constant. This new feature for a symbolic algorithm has a
  huge impact on the practical efficiency as it is illustrated by the
  experiments on classical benchmark polynomial systems, at the end of the
  paper.
\end{abstract}

\section{Introduction}
Solving polynomial systems is the cornerstone of many applications in
domains such as robotics, geometric modeling, signal processing,
chromatology, structural molecular biology etc. In these problems,
the system has, most of the time, finitely many solutions and the
equations often appear with approximate coefficients. 

From a computational point of view, it is an actual challenge to develop
efficient and stable methods to solve such problems. First backward stability
is expected. The computed solutions should be the exact solutions of a system
in the neighborhood of the input system. Efficiency is also mandatory to
tackle the encountered polynomial systems. One can expect for instance that
the behavior of the method depends mainly on the number of solutions, and
partially on other extrinsic parameters such as the number of variables.

To handle the backward stability issue, one may consider classical numerical
methods such as Newton-like iterations. However these local methods do not
provide any guarantee of global convergence nor a complete description of all
the roots. Algebraic methods, on the contrary, handle all the roots
simultaneously. They reduce the problem to computing the structure of the
quotient algebra $\Ac$ of the polynomial ring modulo the ideal $I$ generated
by  the input  system \cite{irsea-em-07}.
Such a structure is given by a basis $B$ of $\Ac$ as a vector
space, and the tables of multiplication in $\Ac$.
Equivalently, it can be described by an algorithm of projection of the
ring of polynomials $\kx$ onto the vector space $\<B\>$ generated by $B$,
along the ideal $I$. We call such a projection, a normal form for the ideal $I$.  

From the knowledge of the multiplication tables, we deduce
either an exact encoding of the roots through a rational univariate
representation %\cite{GRR97},
\cite{Rouillier:99}, \cite{GiLeSa:1999:kro-focm99},
or a numerical approximation by eigenvector
computation, \cite{AuSt88}, \cite{BMr98}, %\cite{GAE99},
\cite{Stetter04}. 
Since the eigencomputation can be considered as a numerically well-controlled
process \cite{GLVL96}, the challenge becomes now to compute efficiently and
in a stable way the quotient algebra structure $\Ac$.

In the family of algebraic methods, resultant-based techniques (see
eg. \cite{Mac02}, \cite{EMres99}, \cite{bem-upoicagd-03}) exploit the
properties of coefficients matrices of monomial multiples of the input equations
in some specific degree. The table of multiplications are obtained by explicit
Schur complements in these matrices \cite{EMres99}. Their construction is deeply linked to
the geometry of the underlying variety that make these methods very tolerant
against small perturbations. Unfortunately they heavily rely on genericity
hypotheses that reduce their applicability. Moreover, the size of the
constructed matrices usually grows exponentially with the number of variables.

To avoid such a pitfall, so-called H-bases have also been studied
\cite{Mac16}, \cite{MoSa02}. They proceed degree by degree and 
stop when the terms of highest degree of the computed polynomials generate
the terms of highest degree of all the ideal $I$. The stopping criterion
requires the computation of generators of the syzygies of these highest
degree terms. Though it also yields a basis of the quotient ring $\Ac$,
without any apriori knowledge on its dimension, practically speaking, it also
suffers from the swelling of the size of the linear systems to be solved.

The approach can be refined further by using a grading for which the highest
term of a polynomial is a monomial. This leads to Gr\"obner basis computation
(see eg. \cite{CLO92}). This approach also yields a basis of the quotient
algebra $\Ac$ and a normal form for $I$. As in the other methods, their computation can be seen as just a
triangulation of a certain matrix. Unfortunately these methods suffer from
unavoidable instability: the monomial ordering attached to the Gr\"obner
basis make the pivot selection strategy in the triangulation depend only
on the symbolic structure of the rows of the matrix, and not on the
numerical values of the coefficients appearing in these rows.  This {\em can}
lead to artificial unwanted singularities in the representation of the
quotient ring (compare for instance the degree reverse lexicographic
Gr\"obner basis of 
$p_1=ax_1^2+bx_2^2+\varepsilon_1x_1x_2$,
$p_2=cx_1^2+dx_2^2+\varepsilon_2x_1x_2$ with $\varepsilon_1=\varepsilon_2=0$
and $\varepsilon_1\not=0$, $\varepsilon_2\not=0$).
%and $f_{1} + \varepsilon_{1} x_{1} x_{2}$, $f_{2}
%+ \varepsilon_{2} x_{1} x_{2}$, where $a,b,c,d \in \kk$, with
%$a\,d-b\, c\not= 0$ and $n_{1}(x_{1},x_{2}), n_{2}(x_{1},x_{2})=0$ are
%linear terms).

To circumvent this artificial difficulty, a new approach based on a normal
form criterion which involves commutation relations (for bases $B$ connected
to $1$) was first proposed in \cite{BMnf99}.  Further investigations of this
technique \cite{BMPhT00}, \cite{MT02synasc}, including the Ph.D. dissertation
\cite{PhT02} lead to a first version
\cite{mt:gnfpss-issac-05} of a normal form algorithm, which allows to construct
efficiently and in a stable way zero-dimensional quotient algebra
representations. Similar investigations were also pursued in
\cite{kkr:aavbb-05}, \cite{kekr:cbb-jpaa-05}, but in the more restrictive
case where the basis is stable by division. Other investigations related to the
stabilization of the normal form process can also be found in \cite{Stetter04}.

In this paper, we discuss border basis algorithms in the general sense,
that is for bases $B$ which are connected to $1$. See \cite{snc:pdsncbbm-06}
for an introductory presentation of their properties.
As in \cite{FGLM93} or \cite{Faug99}, our approach is based on linear
algebra tools. As for H-bases, we use a grading of the polynomial ring, but the
construction is optimized in the spirit of \cite{BMPhT00}. 
We describes an efficient criterion based on commutation polynomials to check
the normal form property. This leads to an algorithm, as presented in
\cite{mt:gnfpss-issac-05}, which has been improved to treat
optimally the case when a syzygy of the components of highest degree is
found, which is not a syzygy of the corresponding polynomials. We prove that
the syzygies of a (general) border basis are generated by the commutation
relations, giving a short and concise answer to a conjecture in
\cite{kekr:cbb-jpaa-05} for basis stable by division.
Meanwhile, works related to this conjecture for this special case 
were also investigated in \cite{Huibregtse06}.
Regarding the numerical stability of border bases, we prove that our
construction of normal form is stable against small perturbations of
the input system, if the number of solutions remains constant. 

The paper is structured as follows. In the next section, we recall the
notations, used in section \ref{sec:normform} to prove the stopping criterion
for generalized normal forms. In section \ref{sec:syzygy}, we show how to
recover the syzygies from the commutation relations.  In section \ref{sec:algo},
we recall briefly the maind idea of the algorithm described in \cite{mt:gnfpss-issac-05}. 
In section \ref{sec:stability}, we analyze the stability of this
algorithm from a symbolic-numeric perspective.  Finally, we show the
efficiency of our implementation and its numerical behavior on classical
polynomial benchmarks.

\section{Notations}

We recall some of the definitions stated in \cite{BMnf99},
\cite{BMPhT00}, \cite{MT02synasc} and add a few more that we will need in
the sequel. 

Let $\kk$ be an effective field. The ring of $n$-variate polynomials over
$\kk$ will be denoted by $R$, $R=\kx=\kk[x_{1},\ldots,x_{n}]$. We consider
$n$-variate polynomials $f_{1},\ldots, f_{s}\in R$. Our goal is to solve the
system of equations $f_{1}=0,\ldots, f_{s}=0$ over the algebraic closure
$\overline{\kk}$ of $\kk$.  These polynomials generate an ideal of $\kx$ that
we call $I$.  The quotient of $\kx$ modulo $I$ will be denoted by $\Ac$. From
now on, we suppose that {\em $I$ is zero dimensional} so that $\Ac$ is a
finite dimensional $\kk$-vector space.  The roots, with coordinates in the
algebraic closure of $\kk$, will be denoted by $\zeta_1,\ldots, \zeta_{d}$,
with $\zeta_{i}=(\zeta_{i,1},\ldots,
\zeta_{i,n}) \in \overline{\kk}^{n}$.

The support ${\rm \supp}(p)$ of a polynomial $p\in \kx$ is the set of
monomials appearing in $p$ with non-zero coefficients. Given a set $S$ of
elements of $\kx$, we denote by $\<S\>$ the $\kk$-vector space spanned by the
elements of $S$.  Finally, we denote the set of all the monomials in the
variables $\x=(x_{1},\ldots,x_{n})$ by $\Mon$. 
A term is an element of the form $\lambda\cdot m$ with $\lambda\in
\kk-\{0\}$ and $m\in \Mon$. For a subset $S$ of $\Mon$,
we will denote by $S^c$ the set-theoretical complement of the set $S$ in
$\Mon$. For any monomial $m\in \Mon$, $m\cdot\Mon$ denotes the set of
all monomial multiples of $m$.

For any subset $S$
of $R$, we denote by $S^+$ or $D(S)$ the set $S^{+}= S \cup x_{1}\, S \cup \cdots \cup
x_{n}\, S$, $\partial S= S^{+} \backslash S$.  
$S^{+}$ is called the {\em prolongation} of $S$.
For any $k\in \nn$, $S^{[k]}$ is $S^{\stackrel{k\ \mathrm{times}}{+\cdots+}}$
the result of applying $k$ times the operator $^{+}$ on $S$. By convention, $S^{[0]}$ is $S$.
% the subset of $S$ of elements of degree at most $k$, and $S_{[k]}$ will 
% denote the set of elements of $S$ of degree exactly $k$.

If $B \subset \Mon$ contains $1$, for any monomial $m \in \Mon$, there exists $k$ such
that $m\in B^{[k]}$. We say that a monomial $m$ is of $B$-index $k$ if $m \in
B^{[k]}-B^{[k-1]}$, and we denote it by $\delta_{B}(m)$.

%\begin{defi}
A set of monomials $B$ is said to be connected to $1$, if and only if,
for every monomial $m$ in $B$, there exists a finite sequence of
variables $(x_{i_{j}})_{j\in [1,l]}$ such that
%\begin{itemize}
%\item 
$1\in B$,
%\item 
$\forall l'\in [1,l],\ \Prod_{j\in[1,l']}x_{i_{j}}\in B$
 and $\Prod_{j\in [1,l]}\, x_{i_{j}}\ =m$.
%\end{itemize}
%\end{defi}
% According to this definition, 
% $\{1,x,x^2,x^2y\}$ is connected to $1$, but
%  $\{1,x,x^2y\}$ is not connected to $1$, since $x^2y$ is isolated from
% $1$. 
%  \begin{defi}

A set of monomials $B$  is stable by division, if for any $m \in
B$ and any variable $x_{i}$ such that $m=x_{i}\, m'$ ($m'\in \Mon$), we have $m'$ in $B$.
%\end{defi}
Remark that a set $B$ stable by division, is connected to $1$.

\begin{defi}
Let $\Lambda$ be a monoid with a well order relation $\prec$, such that:
$$
\forall \alpha, \beta, \gamma \in \Lambda,\ \alpha\prec\beta\Rightarrow 
        \gamma+\alpha\prec \gamma+\beta
$$
A $(\Lambda,\prec)$-grading of $\kx$ is the decomposition of $\kx$ as the
direct sum: 
$\kx=\bigoplus_{\lb \in\Lambda}\kx_{[\lb]}$,
with the following property:
$$\forall f\in\kx_{[\alpha]},\ g\in\kx_{[\beta]}
\Rightarrow f\,g\in\kx_{[\alpha+\beta]}.
$$ 
We denote by degree of $f$ or $\deg_\Lambda(f)$, or $\deg(f)$ (when no
confusion is possible) and by $\Lambda(f)$, the following element of $\Lambda$: 
$$
\Lb(f)=\min\{\lb \in\Lambda\, |\,
f\in \bigoplus_{\lb' \preceq \lb}\kx_{[\lb']}\}.
$$
\end{defi} 
For any set $V \subset \kx$, let $V_{\lb} =\bigoplus_{\lb' \preceq \lb}\ 
\kx_{[\lb']} \cap V$.  For any $\lb \in \Lb$, let $\lb^{+} = \min
\{\lb'\in \Lb;\ \kx_{\lb}^{+} \subset \kx_{\lb'} \}$ and let $\lb^{-}
= \max \{\lb'\in \Lb;\ \kx_{\lb'}^{+} \subset \kx_{\lb} \}.$

In order not to be confused with different notions of degree, for any 
monomial $m$ of $\kx$, we define the size of $m$, denoted by $|m|$, to be the
integer $d$ such that $m=x_{i_1}\cdots x_{i_d}$, which itself imposes a grading.
% Obviously, the size of a monomial coincide with its degree when we
% consider $\kx$ equipped with the standard grading:
% $ 
% \deg_{\NN} (x_{1}^{\alpha_{1}} \cdots  x_{n}^{\alpha_{n}}) = \alpha_{1}
% +\cdots + \alpha_{n}.$
% Then for all $d \in \NN$, $\kx_{[d]}$ is the set of homogeneous polynomials 
% of degree $d$ in the variables $x_{1},\ldots,x_{n}$, $d^{+}=d+1$ and if
% $d>0$, $d^{-}=d-1$  otherwise $d^{-}= -\infty.$ 

Another classical grading is the one associated with a
monomial ordering, where $\Lambda =\NN^{n}$ and $\prec$ is a monomial
order (see \cite[p. 328]{Eis94}) such that for all $\alpha=(\alpha_{1},\ldots,
\alpha_{n}) \in \NN^{n}$, we have
$\kx_{[\alpha]} =\kk\, x_{1}^{\alpha_{1}} \cdots x_{n}^{\alpha_{n}}.$
\begin{defi}\label{def:reducinggrad} 
  We say that $\Lambda$ is a reducing grading if $\Lambda$ is a
  grading and if we have the property: for all monomials
  $m,m'\in \Mon$, such that $m'$ divides strictly $m$,
$\Lb(m')  \prec \Lb(m),\ \Lb(m') \neq \Lb(m).$
\end{defi}
Both  gradings induced by the classical degree and a monomial ordering are
reducing grading.
Hereafter, {\em we will denote by $\Lb$ a reducing grading}.
\begin{defi}
A \rewriting family $F$ for a monomial set $B$ is a set of
polynomials  $F=\{f_i\}_{i\in \mathcal{I}}$ such that 
%\begin{itemize}
%\item 
$\supp(f_i)\subset B^+$,
%\item
$f_i$ has exactly {\bf one} monomial $\gamma(f_i)$ 
(also called the {\it leading monomial} of $f_i$) in  $\partial B$,
%\item 
if  $\gamma(f_i)=\gamma(f_j)$ then $i=j$,
%\end{itemize}
\end{defi}
Remark that the elements of $F$ can be seen as rewriting
rules for the leading monomial using monomials of $B$.
\begin{defi}\label{defnormfam}
A reducing family $F$ of degree $\lb \in \Lb$ for a set $B$ 
%connected to $1$ 
is a set of polynomials such that 
%\begin{itemize}
%\item 
$F$ is a \rewriting family for $B$,
%\item 
$\forall m \in \partial B$ of  degree at most $\lb$, $\exists
f\in F \mid \gamma(f)=m$. 
%\end{itemize}
\end{defi}
{}For the set $B=\{1,x_0,x_1,x_0x_1\}$, the set of polynomials
$F=\{x_0^2-1,x_1^2-x_{1},x_0^2x_1-x_{1},x_1^2x_0-x_{1}\}$ is a reducing
family of degree $3$.

Notice that a reducing family of degree $\lb$ for a set $B$
(connected to $1$) allows to rewrite the monomials of $\<B^+\>_{\lb}$
modulo $F$ as elements of $\<B\>_{\lb}$. This leads in fact, to the
definition of the linear projection $\Prj_{F}$, associated to a reducing
family for a set $B$ connected to $1$.
\begin{defi}
Given a reducing family $F$ of degree $\lb \in \Lambda$ for a set $B$
connected to $1$, we define the linear projection $\Prj_{F}: \<B^{+}\>_{\lb}\rightarrow
\<B\>_{\lb}$ such that
$$
\begin{array}{l}
\forall m\in B_{\lb},  \Prj_{F}(m)=m,\\
\forall m\in \partial B_{\lb},  \Prj_{F}(m)=m-f;
\end{array}
$$
where $f \in F$ is the unique member of $F$ such as $m=\gamma(f)$.
We extend this construction to $\<B^+\>_{\lb}$  by $\kk$-linearity.
\end{defi}

There is a parallel with the differential algebra terminology, which we want
to highlight here. To a polynomial $f(x_{1},\ldots, x_{n}) \in R$ of degree
at most $\delta$, we can associate the differential equation
$f(\partial_{1},\ldots, \partial_{n}) \Phi(t_{1},\ldots,t_{n})=0$, where
$\partial_{i}$ is the derivation with respect to the variable $t_{i}$. It is
a linear equation in $\Phi$. The solution $\Phi(t_{1},\ldots,t_{n})$ lives
in the ring $\kk[[t_{1},\ldots,t_{n}]]$ of formal power series in $t_{1},
\ldots, t_{n}$, which we can truncate in degree 
$\lambda \succeq \Lambda(f_{i})$.
This set $\Jc_{\lb}(\tb)$ of truncated series in degree
$\lb$ is called the space of {\em $\lb$-jets} in the literature. See
eg. \cite{Saunders89} for more details.

Given a set of polynomials $F \subset \kk[\x]$ of degree at most $\lb$, we
consider the so-called {\em solution manifold} $\Rc_{\lb}$ of the
differential system $F(\partial)(\Phi)=0$ in $\Jc_{\lb}(\tb)$ (In our case,
it is just a linear space). Given $\lb'\prec \lb$ and $\Phi \in
\Jc_{\lb}(\tb)$, we can forget the coefficient of degree bigger than $\lb'$
and project it on $\Jc_{\lb'}(\tb)$.  The set of solutions of
$F(\partial)(\Phi)=0$ projected on $\Jc_{\lb'}(\tb)$, is then defined by the
equations $F_{\lb'}(\Phi')=0$ where $\<F_{\lb'}\> = \<F\> \cap
\kk[\x]_{\lb'}$. We denote it by $\Prj_{\lb'}(F)$. 
We denote by $D(F)$ the new system which extends $F$ with the
$\partial_{i}f(\partial)=0$, for $i\in1\ldots n$, $f\in F$, obtained by
formal multiplication by the $\partial_{i}$. It is called the prolongation of $F$.

The system $F$ of degree at most $\lb$ is said {\em formally integrable} if
for any $r\ge 0$, $\Prj_{\lb}(D^{r+1}(F)) = F$. A technical condition of
involutivity is introduced to ensure the regularity of the differential
system \cite{Pommaret94}.  A result of Cartan-Kuranishi
\cite{Cartan45,Kuranishi57,Pommaret94} asserts that any system $F$ can be
transformed by prolongation and projection into a (involutive) formally
integrable system. The connection of this construction with the notion of
Mumford regularity of polynomial systems has been detailed in
\cite{Malgrange03}. This correspondence has been used explicitly for solving
polynomial equations in \cite{ReZhi:icpss:04}.  This involutive division is
also used to construct a family of normal form projection related to
so-called involutive bases (see eg.
\cite{GBBIB1162}). In Janet basis construction, however one difference
is that the variables do not play a symmetric role. So-called
multiplicative variables are used to perform reduction and
non-multiplicative to extend the polynomial vector space.

In the following, instead of working in the vector space of all polynomials
of a given degree, we will consider
completion procedures based on prolongations and projections relative to a
specific set of monomials $B$. Dealing locally with this set of monomials related
to the number of solutions of the system will allow us to improve
significantly the linear algebra stages in this type of methods.

In the sequel, we will make a heavy use of multiplication
operators by one variable that we define as follows:
%\begin{defi}
%We define 
\negskip\begin{eqnarray*}
M_{i,\lb}: \<B\>_{\lb^{-}} & \rightarrow & \<B\>_{\lb}\\
 b & \mapsto & \Prj_{F}(x_i b ).
\end{eqnarray*}
%\end{defi} 
%\begin{remark}
 The subscript $_{\lb}$ is redundant as soon as we know that $F$ is a
 reducing family of degree $\lb$, and we will omit this subscript in the sequel.
%\end{remark}   
\begin{defi}
 Let $F=\{f_1,\dots, f_s\}$ be a polynomial set, we denote by
 $F_{\<\lb\>}$ the vector space: $$F_{\<\lb\>}=\<\{x^\alpha f_i|
\ \Lb(x^\alpha f_i) \le \lb\}\>.$$
\end{defi}
Obviously, we have $F_{\<\lb\>} \subset (F)_{\lb}$ where $(F)$ is the ideal
generated by $F$. 
Next we introduce a definition, which is weakening the
notion of monomial ordering for Gr\"obner basis:

\begin{defi}\label{def4.10}
We say that a function $\gamma: \kx \rightarrow \Mon$ ($\Mon$ is the
set of all monomials in the variables $x_1,\dots,x_n$), is a choice
function refining the grading $\Lambda$, if for any polynomial $p$,
$\gamma(p)$ is a monomial such that
%\begin{itemize}
%\item{} 
$\gamma(p)\in \supp(p)$,
%\item{} 
if $m\in \supp(p)$, $m\not=\gamma(p)$ then $\gamma(p)$ does not
divide $m$, 
%\item{} 
and $\Lambda(\gamma(p))=\max\{\Lambda(m),\ m\in \supp(p)\}$. 
% \item{} if $\forall i\in [1,k] \# max_i(max_{i-1}(\dots max_1(p)\dots ))\geq 2$ then $\gamma(p)\in max_i(max_{i-1}(\dots max_1(p)\dots ))\geq 2$
 %\item{} if $\exists i\in[1,k]\ s.t.\ \# \max_i(max_{i-1}(\dots max_1(p)\dots ))=1$ then $\gamma(p)$ is this element.
%\end{itemize}
The coefficient of the monomial $\gamma(p)$ in $p$ will be denoted by $\kappa(p)$.
\end{defi}
\begin{example}\label{macfun}
In the following, we consider a {\em Macaulay}\footnote{It gives
monomial basis similar to those given by Macaulay in his multivariate resultant
construction \cite{Mac02}.} 
choice function $\gamma$, such that for all $p\in \kx$, $\gamma(p) =
x_{1}^{\alpha_{1}}\cdots$ $x_{n}^{\alpha_{n}}$ satisfies,
$|\gamma(p)|= \max \{ |m|; m \in \supp(p)\}=d$,
and $\exists i_{0}$ st. $\alpha_{i_{0}}= \max \{ \deg_{x_{i}}(m), m
\in \supp(p)$ and $, |m|= d; i=1,\ldots,n\}$. In case more than one
monomial satisfy these conditions, the greatest monomial for the lexicographic
order is chosen.
\end{example}
%\begin{remark}

  The monomial returned by the choice function has
  the same name as the {\it leading} monomial of an element of a
  reducing family. This is intended, as we will define a reducing
  family on the behalf of choice functions, and in this framework the
  two will coincide.
%\end{remark}
Hereafter if $S=\{p_1,\ldots,p_s\}$ is a polynomial set, then we denote
by $\gamma(S)$ the set: $\gamma(S)=\{\gamma(p_1)\dots \gamma(p_s)\}.$
\begin{defi}\label{def:cpol}
Let $\gamma$ be a choice function refining a grading $\Lambda$.
 For any polynomials $p_1,p_2 \in \kx$, let the C-polynomial relative
to $\gamma$ and $(p_1,p_2)$ be
$$
C(p_{1},p_{2})=\frac{\lcm(\gamma(p_1),\gamma(p_2))}{\kappa(p_{1})\,
\gamma(p_1)}p_1 -\frac{\lcm(\gamma(p_1),\gamma(p_2))}{\kappa(p_{2})\,\gamma(p_2)}p_2.
$$ 
Let the C-degree of $(p_{1},p_{2})$ be ${\Lambda}(\lcm(\gamma(p_1)$,
$\gamma(p_2)))$
%%% CHECK $/\gamma(p_1))p_1$,
and let the  leading monomial of the pair $(p_{1},p_{2})$ be $\lcm(\gamma(p_1),\gamma(p_2))$.
\end{defi}
This is almost the same definition as a $S$-polynomial \cite{CLO92} when
$\gamma$ is a monomial ordering. We however use a new name to underline that
now $\gamma$ may not be a monomial ordering (i.e. a total order 
compatible with monomial multiplication). As we will see in the next section,
the $C$-polynomials express commutation conditions for the
$M_{i,\lb}$.
%Indeed, if up to degree $\lb$ the C-polynomials of a given
%reducing family of degree at most $\lb$ vanish, we will see that the
%$M_{i,\lb}$ are pairwise commuting. 

\section{Generalized normal form  criterion}\label{sec:normform}
Let $F=\{f_1,\ldots,f_s\}$ be a polynomial system and let $I$ be the ideal
generated by $F$. Remember that $F_{\<\lb\>}$ (the
$\kk$-vector space spanned by the monomial multiples of the $f_i$, $x^\alpha
f_i$ of degree $\preceq \lb \in \Lb$) is included in $I_{\lb}$. Thus,  when $I_{\lb} =F_{\<\lb\>}$
we can define a normal form modulo $I$, up to the degree $\lb$ as the
projection of $\kx_{\lb}$ along $F_{\<\lb\>}$ onto a supplementary space
$\<B\>_{\lb}$. Hereafter, we consider a set $B$ of monomials, containing $1$.

Let $F$ be a \rewriting family, and let $\Hc=\{m,\ \exists p\in
F,\gamma(p)=m\}$  be the set of their leading
monomials then, obviously $F$
allows us to define the projection $\Prj_F$ of $B\cup\Hc$ on $B$ along
$\<F\>$. However we may extend this projection using the following extension
process:
\begin{defi}
  Let $F$ be a \rewriting family. For all $m\in B$, we define $\Prje_F(m)=m$. 
For $m \not\in B$, there exists $m'\in \partial B$ and $r$ integers
  $i_1,\ldots,i_r\in [1,n]$, such that $m=x_{i_1}\cdots x_{i_r}m'$.  We define
  $\Prje_F(m)$ by induction on $k$, as follows.\\
  $\bullet$ if $r=0$, $\Prje_F(m')$ is defined as $\Prje_{F}(m')=\Prj_F(m')=m'-f$
    where $f\in F$ is such that $\gamma(f)=m'$.\\
$\bullet$ $\forall k, 1\leq k\leq r$, $\Prje_F(x_{i_{r-k}}\cdots
    x_{i_r}m')=\Prj_F(x_{i_{r-k}}\Prje_F(x_{i_{r-k+1}}\cdots x_{i_{r}}m'))$, if
    this latter quantity is defined. Otherwise we say that $\Prje_{F}(m)$
    is undefined.
\end{defi} 
Remark that the above process allows us to define $\Prje_{F}$ only on 
monomials, and we extend it implicitly, by linearity.
 Remark also that this
extension process  is not defined in a unique way. Indeed, two different
decompositions of a monomial $m$ may lead to two different values of 
$\Prje_{F}(m)$. However the following theorem shows that this extension process
becomes canonical as soon as we check some commutativity conditions.
\begin{theorem}\label{thm52}
  Assume that $B$ is connected to $1$.  Let $F$ be a \rewriting
  family, and let $E$ be the set of monomials $m$ such that for all
  decomposition of $m$ as a product of variables, $m=x_{i_0}\cdots
  x_{i_k}$, $\Prj_F(x_{i_0}\Prj_F(x_{i_{1}} \cdots \Prj_F(x_{i_k})\cdots))$ is defined.
  Suppose that for all $m\in E$ and all indexes $i,j\in[1,n]$ such that
  $x_ix_jm\in E$, we have:
  $$ \Prje_{F}(x_i\Prje_{F}(x_jm))=\Prje_{F}(x_j\Prje_{F}(x_im)). $$
Then $\Prje_{F}$ coincides with the linear projection $\Prj_{S}$ of $\<E\>$ on $\<B\>$
along the vector space spanned by the polynomials $S=\{x^\alpha f,\ 
  \alpha\in \nn^n,\ f\in F \mbox{ and } x^\alpha\gamma(f)\in E\}$.
\end{theorem}
\begin{proof}
Remark that the way we define it, makes $\Prje_F$ inherently a {\em linear}
multivalued map. Hence to prove the theorem we have first to show that
under the above hypotheses, $\Prje_{F}$ becomes a well defined map, and
next that this well defined linear map coincide with the projection
$\Prj_{S}$ of $\<E\>$ on $\<B\>$ along $\<S\>$.

Remark also that $E$ is obviously stable by monomial division: if all the
possible decompositions of $m$ as a product of variables $m=x_{i_0}\cdots
x_{i_k}$ are such that $\Prj_F(x_{i_0}\cdots$ $\Prj_F(x_{i_k}))$ is defined,
then a fortiori if $m'$ is a divisor of $m$, this property is true for $m'$.
Let us show that the extension process is independent of the way $m$ is
decomposed as a product of variables. Let $m=x_{i_{0}}\, m'
=x_{i_{1}}\, m''$ with $i_{0}\neq i_{1}$ and $m, m', m'' \in E$, then there
exists $m''' \in E$ (since $E$ is stable by monomial division) such that $m=
x_{i_{0}}\, x_{i_{1}} m'''$. As $m,\ m',\ m'',$ and $m'''$ are in $E$,
$\Prj_{F}^e(m')$, $\Prj_{F}^e(m'')$, $\Prj_{F}^e(x_{i_0}m')$,
$\Prj_{F}^e(x_{i_1}m''),$ and $\Prj_{F}^e(m''')$ are defined and we have: $$
\begin{array}{rcl}
\Prje_{F} (x_{i_{0}}\, \Prje_{F}(m')) 
 & = & \Prje_{F}(x_{i_{0}}\, \Prje_{F}(x_{i_{1}} \Prje_{F}(m''') )), \\
\Prje_{F} (x_{i_{1}}\, \Prje_{F}(m'')) 
 & = & \Prje_{F}(x_{i_{1}}\, \Prje_{F}(x_{i_{0}} \Prje_{F}(m'''))).
\end{array}
$$
The commutation condition guarantees that the two quantities are
equal, so that the definition of $\Prje_{F}$ does not depend on the way
to write $m$ as a product of variables.\\
%Let us remark here that obviously $\tilde{R}$ and $R$ coincide on $\<B^+\>$.
Next we have to show that $\Prje_{F}$ and $\Prj_{S}$ coincide on their common
set of definition. We do it by induction on the size of the monomials:\\
It is true that $\Prje_{F}(1)=\Prj_{S}(1)=1$ (since $1 \in B$). For any monomial
$m\not=1$ in $E$, the property of connectivity of $B$ and the definition
of $E$ gives us: $\exists m'\in E\mbox{ and } i_0\in [1,n] \mbox{ such
  that } m=x_{i_0}m'$ and $\Prje_{F}(m')$ is defined, so that we have:
\negskip\begin{eqnarray*}
\Prje_{F}(m)&=&\Prje_{F}(x_{i_0}m')=^{def}\Prje_F(x_{i_0}\Prje_{F}(m'))
 = ^{induction}\Prje_{F}(x_{i_0} \Prj_{S}(m')) \in \<B\>.
\end{eqnarray*}
Now by induction, $m'- \Prj_{S}(m') \in S_{\lb^{-}}$ where $\lb=\Lambda(m)$ and 
\negskip\begin{eqnarray*}
{m-\Prje_{F}(m)= x_{i_{0}}(m'- \Prj_{S}(m'))}
 +  \left( x_{i_{0}}\Prj_{S}(m')-\Prje_{F}(x_{i_0}\Prj_{S}(m')) \right) \in \<S\>. 
\end{eqnarray*} 
Thus $\Prje_{F}(m)$ is the projection of $m$ on $\<B\>$ along $\<S\>$.
\end{proof}

Suppose now that we are given a reducing family of degree $\delta$
instead of a \rewriting family. Then we can further extend the
above theorem with the help of the following lemma.
\begin{lemme}\label{lemme:normformdef}
  Let $F$ be a reducing family of degree $\lb$ for a set $B$ connected
  to $1$, and suppose that $\forall f\in F,\ \Lb(\gamma(f))=\Lb(f)$.
  With the notation of Theorem \ref{thm52}, the set $E$ contains the set of
monomials of degree less than or equal to $\lb$.
\end{lemme}
\begin{proof}
Let $m\in\mc_\lb$ be a monomial of degree less or equal to $\lb$, then $m$
can be written as $m=x_{i_1}\ldots x_{i_d}$ with $d=|m|$.  
% However the fact
% that $\delta$ is a reducing grading implies that $\forall k\leq d,\
% {\Lb}(x_{i_1}\ldots x_{i_k})\leq{\Lb}(x_{i_1}\ldots
% x_{i_d})$. Recall now that for computing $\Prje_{F}(m)$, we apply successively
% $\Prj_F$ to monomials of $B^+$. Precisely we compute $\Prj_F(x_{i_d}\Prj_F(\cdots
% \Prj_F(x_{i_1}))\dots)$.
Let us prove by induction on $k\leq d$, that 
$p_k=\Prj_F(x_{i_k}\Prj_F(\cdots \Prj_F(x_{i_1}))\cdots)$ is defined and that
${\Lb}(p_k)\le {\Lb}(x_{i_k}\cdots x_{i_1})$.\\  
%  Then $x_{i_{k+1}}p_k$ is composed of monomials belonging to $B^+$
%  and of degree at most $\Lambda(x_{i_{k+1}}\cdots x_{i_1})$. 
Consider now $x_{i_{k+1}}p_{k}\in B^{+}$. For $m'\in
\supp(x_{i_{k+1}}p_{k})\cap B$, we have $\Prj(m')=m'$ and 
by induction hypothesis $\Lambda(m') \le \Lambda(x_{i_{k+1}}\cdots
x_{i_1})$. 
For $m' \in \supp(x_{i_{k+1}}p_k) \cap\partial B$, as $F$ is a
reducing family of degree $\lb$, we have a rewriting rule for $m'$. The
hypothesis that $\forall f\in F,\ \Lb(\gamma(f))=\Lb(f)$ implies that $m'$
rewrites in terms of monomials of degree bounded by
$\Lambda(x_{i_{k+1}}\cdots x_{i_1})$. This proves $p_{k+1}:=
\Prj_{F}(x_{i_{k+1}} p_{k})$ is defined and ${\Lb}(p_{k+1}) \le
{\Lb}(x_{i_{k+1}}\cdots x_{i_1})$.\\ 
This proves by induction that 
$\Prje_{F}(x_{i_{1}}\cdots \Prje_{F}(x_{i_d}))$ is defined, for any
decomposition $m =x_{i_{1}}\cdots x_{i_{d}} \in \Mon_{\lb}$ so that $m\in E$.
%  Of course, for all $m\in \supp(x_{i_{k+1}}p_k)\cap B$,
%  $\Lambda(m)\le \Lambda(x_{i_{k+1}}\cdots x_{i_1})$.
This ends the proof.
\end{proof}
\begin{theorem}\label{thm1}
Let $F$ be a reducing family of degree $\lb$ for a set $B$
connected to $1$.  If we have: 
\begin{itemize}
 \item $\forall f\in F,\ \Lb(\gamma(f))=\Lb(f)$.   
 \item
    $M_{j,\lb}\circ M_{i,\lb^{-}}=M_{i,\lb} \circ M_{j,\lb^{-}}$, for $1\leq
    i,j\leq n$,
\end{itemize} 
then, we can extend $\Prj_{F}$ to a linear projection  $\Prje_{F}$ from
  $\kx_{\lb}$ onto $\<B\>_{\lb}$ with kernel  $F_{\<\lb\>}$. 
\end{theorem}
\begin{proof}
  As $F$ is a reducing family of degree $\lb$, by Lemma
  \ref{lemme:normformdef}, we have $E\supset \Mon_{\lb}$.
%, which is stable by monomial, division.
  
Let us prove that for all $m \in \Mon_{\lb^{--}}$ and all pairs of indices
$(i,j)$, there exists a way to define $\Prje_{F}$ such that
$\Prje_{F}(x_i$ $\,\Prje_{F}(x_j m))=\Prje_{F}(x_j \, \Prje_{F}(x_im)).$ 

{}As $\Mon_{\lb^{--}} \subset \Mon_{\lb} \subset E$, $\Prje_{F}(m)$ is
defined and $\supp(\Prje_{F}(m))\subset B$. 
We define $\Prje_{F}(x_im)=\Prj_F(x_{i}\Prje_{F}(m))$ and, similarly,
$\Prje_{F}(x_jm)=\Prj_F(x_j\Prje_{F}(m))$. With this definition we have: 
\negskip\begin{eqnarray*}
\Prje_{F}(x_{i}\Prje_{F}(x_jm)) &=& M_{i,\lb}(M_{j,\lb^{-}}(\Prje_{F}(m))) \\
& =&M_{j,\lb}(M_{i,\lb^{-}}(\Prje_{F}(m)))=\Prje_{F}(x_j \Prje_{F}(x_im)). 
\end{eqnarray*}
which proves the commutation property.     
We end the proof by applying Theorem \ref{thm52}.
%By \ref{thm52}, to prove this theorem, we just
%  have to show that we are able to define $\Prje_{F}$ on all the monomials
%  of degree at most $\delta$. Let $m$ be a monomial of degree at most
%  $\ld$ $m$ has a finite size, let us call it $d$: we can write
%  $m$ as, $m=x_{i_1}\cdots x_{i_d}$. We are able to construct
%  $\Prj_F(1)$. We are also able to construct $\Prj_F(x_{i_1})$: $\Prj_F$
%  exists for all monomials of $B^+$ of degree less than $\lb$ and
%  $x_{i_1}$ is element of $B^+$ and of degree less than $\lb$ as
%  $x_{i_1}$ divides $m$. The same arguments show that we can define
%  $\Prje_{F}(x_{i_2}\Prje_{F}(x_{i_1}))$, etc.  Finally this shows that we are
%  able to define $\Prje_{F}(m)$.
\end{proof}
%We directly deduce from the preceding result, the following property:
\begin{corol}\label{corsum}
With the hypothesis of Theorem \ref{thm1}, we have 
$\kx_{\lb}=\<B\>_{\lb}\oplus F_{\<\lb\>}$. 
\end{corol}
{}Let us give here another effective way to check that we
have a projection from $F_{\<\lb\>}$ (vector space spanned by the monomial
multiples of the $f_i$ of degree $\lb$) onto $\<B\>_\lb$ (element of degree
$\lb$ of the vector space spanned by $B$) starting from a reducing family of
degree $\lb$, without computing {\it explicitly} the multiplication
operators.
\begin{theorem}\label{thm54}
Let $\lb \in \Lb$.  Let $F$ be a reducing family of degree $\lb \in \Lb$,
for $B$. Assume that $\forall\ f\in F,\ \Lb(\gamma(f))=\Lb(f)$ and let
$\Prj_{F}$ be the induced projection from $\<B^{+}\>_{\lb}$ onto
$\<B\>_{\lb}$. Then
$\forall f, f' \in F_{\<\lb\>}$ such that $C(f,f') \in \<B^{+}\>_{\lb}$,
$$
\Prj_{F}(C(f,f')) = 0
$$ 
iff  $\Prj_{F}$  extends uniquely as a projection $\Prje_{F}$ 
from $\kx_{\lb}$ onto $\<B\>_{\lb}$ such that $\ker(\Prje_{F})=F_{\<\lb\>}$.
\end{theorem}
\begin{proof}
By Theorem \ref{thm1}, we have  to show that this condition
is equivalent to the commutation of the operators $M_{i,\lb'},\lb'< \lb$ on the
monomials of $B_{\lb^{--}}$.\\
For any $m \in B_{\lb^{--}}$ and any $i_{1}\neq i_{2}$ such that
$x_{i_{1}}\,m \in \partial B$, $x_{i_{2}}\,m \in \partial B$,
there exists $f, f' \in F_{\<\lb^{-}\>}$ such that
$\gamma(f)=x_{i_{1}}\,m$,  $\gamma(f')=x_{i_{2}}\,m$.
Thus, we have $\Prj_{F}(x_{i_{1}}\,m) = \gamma(f)-f$, $\Prj_{F}(x_{i_{2}}\,m) =
\gamma(f')-f'$ and $C(f,f') = x_{i_{2}}\, f- x_{i_{1}}\, f' \in
\<B\>^{+}_{\lb}$. Consequently, 
\begin{eqnarray*}
\lefteqn{M_{i_2,\lb}(M_{i_1,\lb^{-}}(m))-M_{i_1,\lb}(M_{i_2,\lb^{-}}(m))} \\
& = & M_{i_2,\lb}(\gamma(f)-f)-M_{i_1,\lb}(\gamma(f')-f')\\
& = &
\Prj_{F}(x_{i_2}\gamma(f)-x_{i_2}f)-\Prj_{F}(x_{i_1}\gamma(f')-x_{i_1}f')\\
& = & \Prj_{F}(x_{i_1}f'-x_{i_2}f)=\Prj_{F}(C(f',f)).
\end{eqnarray*}
which is zero by hypothesis. A similar proof applies if $x_{i_{1}}\,m
\in B$ or $x_{i_{2}}\,m \in B$.

Conversely, since $\ker(\Prje_{F})= F_{\<\lb\>}$ and $C(f,f') \in
F_{\<\lb\>} \cap \<B\>^{+}_{\lb}$, we have that
$\Prj_{F}(C(f',f))=\Prje_{F}(C(f',f))=0$, which proves the equivalence
and Theorem \ref{thm54}.
\end{proof}
\begin{remark}\label{remcom}
  In this proof, we  have shown that if the $C$-polynomials up to the
  degree $\lb$ reduce to $0$, then the multiplication operators
  $M_{i,\lb}$ commute.
\end{remark}
\begin{defi}
A reducing family $F$ for all degrees $\lb \in \Lb$ on a set $B$ of
monomials, connected to $1$ will be called a {\em border basis} for $B$.
\end{defi}
Finally, the previous results leads to a new  proof of Theorem 3.1 of
\cite{BMnf99}:
%which says that given a reducing family $F$,  a normal form  $\Prj_{F}$ is
% obtained if and only if the operators $M_i$ commute:  
\begin{theorem}\label{thm:commute}
Let $F$ be a border basis for a set $B$ of monomials, connected to $1$, let
$\Prj_{F}$ be the corresponding reduction from $\<B^{+}\>$ onto $\<B\>$,  
and let $M_{i}: \<B\> \rightarrow \<B\>$ such that $\forall\, b\in \<B\>$,
$M_{i}(b) = \Prj_{F}(x_{i}\, b)$.
Then,
$$ 
M_{j}\circ M_{i}=M_{i}\circ M_{j},\ {\rm for}\  1\leq i,j\leq n
$$
iff there exists a unique projection $\Prje_{F}$ from $\kk[\x]$ onto $\<B\>$
such that $\ker(\Prje_{F})=(F)$ and $(\Prje_{F})_{|\<B^{+}\>}=\Prj_{F}$.
\end{theorem}
\begin{proof}
Under these hypotheses, by Theorem \ref{thm1},  for any $\lb \in \Lb$,
$(\Prj_{F})_{|\<B^{+}\>_{\lb}}$ extends uniquely to a projection
$\Prje_{F_{\lb}}$ from $\kx_{\lb}$ onto $\<B\>_{\lb}$, such 
that $\ker(\Prje_{F_{\lb}})=F_{\<\lb\>}$.
Since for any $\lb,\lb' \in \Lb$ such that $\lb \prec \lb'$, 
we have $(B_{\lb'})_{\lb}
= B_{\lb}$, and $F_{\<\lb\>} \subset (F_{\<\lb'\>})_{\lb}$,
we also have
% $F_{\<\lb\>} = (F_{\<\lb'\>})_{\lb}$ so that 
$(\Prje_{F_{\lb'}})_{|\kx_{\lb}} = \Prje_{F_{\lb}}$. 
This defines a unique linear operator $\Prje_{F}$ on $\kx$ such that 
$\Prje_{|\kx_{\lb}} =\Prje_{F_{\lb}}$ and 
$\ker(\Prje_{F})=\sum_{\lb\in \Lb} F_{\<\lb\>} = (F)$. It proves the direct
implication. The converse implication is immediate.
\end{proof}

\section{Syzygies and Commutation Relations}\label{sec:syzygy}
In this section, we analyze more precisely the relations between the polynomials
of the border basis $F=(f_{\omega})_{\omega \in \partial B}$ where $f_{\omega}=\omega -
\rho_{\omega}$ with $\rho_{\omega} \in \<B\>$.
These relations or syzygies form a module that we denote by
$$
Syz(F) =\{ \sum_{\omega} h_{\omega} e_{\omega} \in \kx^{\partial B}; \sum_{\omega} h_{\omega} f_{\omega}=0\},
$$
where $(e_{\omega})_{\omega\in \partial B}$ is the canonical basis of
$\kx^{\partial B}$. If $F$ is a border basis family constructed from an initial
set of polynomials $H=\{h_{1},\ldots, h_{s}\}$, we can express $f_{\omega}\in F$ in terms of the
polynomials in $H$ (and conversely) so that a syzygy on $F$ induces a syzygy on $H$ (and
conversely). Therefore, we are going here to consider only the syzygies on
$F$.

For any $b \in \<B\>$ and $i\in [1,n]$, $x_i\, b \in \<B^{+}\>$ can be projected in $\<B\>$ along $F$:
$$
\Prj_{F}(x_i\, b) =x_{i}\, b -\sum_{\omega \in \partial B} \mu^{i}_{b,\omega} f_{\omega}.
$$
More generally, for any $p\in \kk[\x]$, we denote by $\mu^{i}_{p,\omega}$ the
coefficient $\mu^{i}_{b,\omega}$ in $\Prj_{F}(x_i\, b)$ for $b=\Prj_{F}(p)
\in \<B\>$ and by $\mu^{i}(p)=\sum_{\omega \in \partial B} \mu^{i}_{p}
f_{\omega}$. Notice that if $x_{i}p\in B$ then $\mu^{i}(p)=0$.

For any $m=x_{i_3}\cdots x_{i_k}\in B$ and $i_{1}, i_{2} \in 1\ldots n$,
the two decompositions 
$\Prj_{F}(x_{i_1} \Prj_{F}(x_{i_2} m)) = \Prj_{F}( x_{i_2} \Prj_{F}(x_{i_1} m))$ yield the syzygy
\begin{equation}\label{gen:syz}
x_{i_1} \sum_{\omega\in \partial B} \mu^{i_{2}}_{m,\omega} f_{\omega}
-
x_{i_2} \sum_{\omega\in \partial B} \mu^{i_{1}}_{m,\omega} f_{\omega} 
- \sum_{\omega\in \partial B} (\mu^{i_{2}}_{x_{i_{1}}
m,\omega}-\mu^{i_{1}}_{x_{i_2} m, \omega})\, f_{\omega} = 0.
\end{equation}
These relations can also be rewritten as:
$$ 
x_{i_{1}}\mu^{i_{2}}(m) + \mu^{i_{1}}(x_{i_{2}}m)
- x_{i_{2}}\mu^{i_{1}}(m) - \mu^{i_{2}}(x_{i_{1}}m) 
=0.
$$
We denote by $\Xi$ the module of $\kx^{\partial B}$ generated by these syzygies. 
It is also the module generated by the relations of commutation
$M_{i_{1}}\circ M_{i_{2}}(b) -M_{i_{2}} \circ M_{i_{1}}(b)=0$ for all $b\in
B$, $i_{i}, i_{2} \in [1\ldots n]$. 
We distinguish the following relations:
\begin{itemize}
 \item If  $x_{i_{1}} m, x_{i_2}m \in B$, then
$\mu^{i_{1}}(m)=\mu^{i_{2}}(m)=0$,
$\mu^{i_{1}}(x_{i_{2}}\,m)=\mu^{i_{2}}(x_{i-1}\,m)=0$
and $m$ does not yield a non-trivial relation of the form \eqref{gen:syz}.
 \item If $x_{i_{1}} m \in B$ but $x_{i_{2}} m \in \partial B$, and 
$x_{i_{1}}\, x_{i_{2}} \, m \in \partial B$,  then 
$\Prj_{F}(x_{i_{2}} m) = x_{i_{2}} m -f_{x_{i_{2}} m}$ and we have the relation
\begin{equation} \label{syz:next-door}
x_{i_{1}} \, f_{x_{i_{2}} m} -f_{x_{i_{1}} x_{i_{2}} m} +
\sum_{\omega\in \partial B} \mu^{i_{1}}_{x_{i_{2}}m,\omega} f_{\omega} =0.
\end{equation}
 \item If $x_{i_{1}} m \in B$, $x_{i_{2}} m \in \partial B$, and 
$x_{i_{1}}\, x_{i_{2}} \, m \in B$,  then 
$\Prj_{F}(x_{i_{2}} m) = x_{i_{2}} m -f_{x_{i_{2}} m}$ and we have the relation
\begin{equation} \label{syz:non-stair}
x_{i_{1}} \, f_{x_{i_{2}} m} +
\sum_{\omega\in \partial B} \mu^{i_{1}}_{x_{i_{2}}m,\omega} f_{\omega} =0.
\end{equation}
 \item If $x_{i_{1}} m \in \partial B$ and $x_{i_{2}} m \in \partial B$, then
$\Prj_{F}(x_{i_{1}} m) = x_{i_{1}} m -f_{x_{i_{1}} m}$, $\Prj_{F}(x_{i_{2}} m) =x_{i_{2}} m -f_{x_{i_{2}} m}$
and we have the syzygy
\begin{equation} \label{syz:across-street}
x_{i_{1}} f_{x_{i_{2}}\, m} - x_{i_{2}} f_{x_{i_{1}} m} -
\sum_{\omega\in \partial B} (\mu^{i_{2}}_{x_{i_{1}} m ,\omega} - \mu^{i_{1}}_{x_{i_{2}} m,\omega}) f_{\omega}=0
\end{equation}
\end{itemize}
The syzygies \eqref{syz:next-door} and \eqref{syz:across-street} are called
respectively {\em next-door} and {\em across-the-street
relations} in \cite{kekr:cbb-jpaa-05}. The syzygies  \eqref{syz:non-stair}
does not exist if $B$ is stable under division (since in this case,
$x_{i_{2}} m  \not\in B$ implies $x_{i_{1}}\,x_{i_{2}} m \not\in B$).
All these syzygies are simply the non-trivial relations induced by the ``border''
$C$-polynomials (see Definition \ref{def:cpol}). 

Let us now prove that the module of syzygies is generated by these
commutations syzygies.
For any monomial $m=x_{i_{1}}\cdots x_{i_{k}}$, we can rewrite by induction
its projection as:
\begin{equation}\label{eq:decomp}
\Prj_{F}(m)= \Prj_{F}(x_{i_1} \Prj_{F}(x_{i_2} \cdots )= m - \sum_{l=1\ldots k} 
x_{i_1}\cdots x_{i_{l-1}} \sum_{\omega\in \partial B} \mu^{i_{l}}_{x_{i_{l+1}}\cdots x_{k},\omega} f_{\omega}
\end{equation}
with the convention that $x_{i_{k+1}}=1$. 
We denote by 
$$
\Xi_{x_{i_1},\ldots, x_{i_{k}}}=\sum_{l=1\ldots k} x_{i_1}\cdots
x_{i_{l-1}} \sum_{\omega\in \partial B} \mu^{i_{l}}_{x_{i_{l+1}}\cdots x_{k},\omega} e_{\omega}
$$  
the corresponding element of $\kx^{\partial B}$.
Notice that this decomposition depends on the order of the decomposition $m=x_{i_1} \cdots x_{i_k}$ as a product of variables.

\begin{lemme}\label{lemme:permute}
If $m=x_{i_1}\cdots x_{i_k}= x_{j_1} \cdots x_{j_k}$ then 
$$
\Xi_{x_{i_1}\cdots x_{i_k}}-\Xi_{x_{j_1}, \ldots, x_{j_k}} \in \Xi.
$$
\end{lemme}
\begin{proof}
Consider first a permutation of two variables  $m=m_1 x_{i_{l}} x_{i_{l+1}} m_2=  m_1 x_{i_{l+1}} x_{i_{l}} m_2$
(with $m_{1}= x_{i_{1}}\cdots x_{i_{l-1}}$, $m_2=x_{i_{l+2}}\cdots x_{i_{k}}$).
Using \eqref{eq:decomp}, the two way of projecting $m=m_1 x_{i_{l}}
x_{i_{l+1}} m_2=  m_1 x_{i_{l+1}} x_{i_{l}} m_2$ yield the syzygy 
\negskip\begin{eqnarray*} 
\lefteqn{\Xi_{\ldots, x_{i_{l}},x_{i_{l+1}},\ldots } -\Xi_{\ldots,
x_{i_{l+1}}, x_{i_{l}},\ldots} = m_{1} \ \times}\\
&& 
\left( x_{i_{l}} \sum_{\omega\in \partial B} \mu^{i_{l+1}}_{m_2,\omega} e_{\omega}
-
x_{i_{l+1}} \sum_{\omega\in \partial B} \mu^{i_{l}}_{m_2,\omega} e_{\omega}
+ \sum_{\omega\in \partial B} ( \mu^{i_{l}}_{x_{i_{l+1}} m_2,\omega} -\mu^{i_{l+1}}_{x_{i_{l}} m_2,\omega} ) e_{\omega}\right)
\end{eqnarray*}
which is in $\Xi$. By iterated permutations of two variables, we transform the sequence $x_{i_1},\ldots, x_{i_k}$ into the sequence
$x_{j_1},\ldots, x_{j_k}$. This allows us to rewrite 
$\Xi_{x_{i_1}\cdots x_{i_k}}$ into $\Xi_{x_{j_1}, \ldots, x_{j_k}}$
modulo $\Xi$.
\end{proof}

\begin{lemme}\label{lemme:cpol}
 If $m \, \theta = m'\, \theta'$ with $\theta, \theta'\in \partial B$ and
$m,m'\in \Mon$, then   
$$
m\, e_{\theta} \equiv m'\, e_{\theta'} + \sum_{\omega \in \partial B} p_{\omega}\, e_{\omega} \ {\rm{}mod}\ \Xi
$$ 
with $p_{\omega} \in \kk[\x]$ of degree $< \max(|m|,|m'|)$.
\end{lemme}
\begin{proof}
If $m=x_{i_1} \cdots x_{i_{d}}$ with $\theta = x_{i_{d+1}} b \in \partial B$ and $b\in B$, the projection formula \eqref{eq:decomp}
of $m\, \theta$ has the form
$$
\Prj_{F}(m\, \theta)= m\, \theta - x_{i_1} \cdots x_{i_{d}} f_{\theta} -
\sum_{l=1\ldots d} x_{i_1}\cdots x_{i_{l-1}} \sum_{\omega\in \partial B} \mu^{i_{l}}_{x_{i_{l+1}}\cdots x_{i_{d}} \theta ,\omega} f_{\omega}
$$
since $x_{i_{d+1}}\, b= \theta$ and $\Prj_{F}(\theta)= \theta -f_{\theta}$. 
If $m'=x_{j_{1}} \cdots x_{j_{d'}}$ and $\theta' \in \partial B$ with $m \theta = m'\theta'$, 
by Lemma  \ref{lemme:permute}, the two decompositions yield the syzygy
\negskip\begin{eqnarray*} 
\lefteqn{x_{i_1} \cdots x_{i_{d}} e_{\theta} - x_{j_1} \cdots x_{j_{d'}} e_{\theta'}}\\
&+& \sum_{l=1\ldots d} x_{i_1}\cdots x_{i_{l-1}} \sum_{\omega\in \partial B}
\mu^{i_{l}}_{x_{i_{l+1}}\cdots x_{i_{d}} \theta ,\omega} e_{\omega} \\
&-&\sum_{l=1\ldots d'} x_{j_1}\cdots x_{j_{l-1}} \sum_{\omega\in \partial B} \mu^{j_{l}}_{x_{j_{l+1}}\cdots x_{j_{d'}} \theta' ,\omega} e_{\omega}
\end{eqnarray*}
as an element of $\Xi$. 
This syzygy is of the form $m\, e_{\theta} - m'\, e_{\theta'} + \sum_{\omega \in \partial B} p_{\omega}\, e_{\omega}$
with $\deg(p_{\omega}) < \max(d,d')$, which proves the lemma. 
\end{proof}

This yields the following result for a border basis, conjectured in
\cite{kekr:cbb-jpaa-05} for the case where $B$ is stable by division:
\begin{theorem} Let $B \subset \Mon$ be connected to $1$
and let $F=(f_{\omega})_{\omega\in \partial B}$ be a border basis for
$B$. Then $Syz(F)$ is generated by the relations \eqref{syz:next-door}, \eqref{syz:non-stair} and
\eqref{syz:across-street}. 
\end{theorem}
\begin{proof}
Let $\sigma = \sum_{\omega} p_{\omega}\, e_{\omega} \in Syz(F)$. 
Consider in this sum, any term $\lambda\, m\, e_{\theta}$  where $\lambda \in \kk-\{0\}$, 
$m \in \Mon$, $\theta \in \partial B$ and $\delta_{B}(m\, \theta)< |m|+1$. 
Then, there exists $m' \in \Mon$, $\theta' \in \partial B$ such that 
$m\,\theta= m'\theta'$ and
$\delta_{B}(m'\, \theta')= |m'|+1$. This implies that $|m|> |m'|$.
Applying lemma \ref{lemme:cpol}, the product $m\, e_{\theta}$ can be reduced modulo $\Xi$, to
$$
m' e_{\theta'} + \sum_{\omega \in \partial B} q_{\omega}\, e_{\omega},
$$
with $\deg(q_{\omega})< |m|$.
Iterating this reduction, we may assume that each term 
$\lambda\, m\, e_{\omega}$ ($\lambda \in \kk-\{0\}$, $m\in \supp(p_{\omega})$, $\omega \in \partial B$)
is such that $\delta_B(m \, \omega)= |m| +1 $. Notice that the polynomial
$m \,f_{\omega}$ has only one monomial of maximal $B$-index, which is $m \, \omega$. 

As $\sum_{\omega} p_{\omega}\, f_{\omega} =0$, there exist $\theta \neq
\theta' \in \partial B$ and monomials $m \in \supp(p_{\theta})$, $m' \in \supp(p_{\theta'})$
such that $m \, \theta = m'\, \theta'$ and $\delta_{B}(m \, \theta)$ is maximal among all terms of the syzygy.
By lemma \ref{lemme:cpol}, we can replace
this pair  $(m \, e_{\theta} , m'\, e_{\theta'})$ modulo $\Xi$ by 
a sum of terms of smaller degree.
This transformation reduces either the number of terms with maximal $B$-index
or the degree of the polynomials $p_{\omega}$.

Since we cannot iterate it infinitely, we deduce that
$\sigma$ is in $\Xi$.
%%generated by the relations \eqref{syz:next-door} and \eqref{syz:across-street}.
\end{proof}

\section{Algorithmic issues}\label{sec:algo}

From the preceding sections, we deduce an algorithm as it is done
in \cite{mt:gnfpss-issac-05}. The main idea is to translate the previous
concepts into linear algebra. From Section \ref{sec:normform} 
to compute effectively a normal form, one has to find a (monomial) basis $B$
of the quotient algebra connected to $1$, and border
relations such that the multiplication operators commute.  The
algorithm described in \cite{mt:gnfpss-issac-05}, is a fix point
method, which updates
\begin{itemize}
 \item a potential monomial basis $B$ of the quotient algebra,
 \item a set $P$ of polynomials or rewriting rules (with one monomial of
their support in $\partial B$ and the remaining monomials in $B$),
\end{itemize}
until a fix point is reached.   
At each step of the algorithm, the following operations are performed:
\begin{enumerate}
 \item The set $P'$ of polynomials of $P^{+}$ with support in $B^{+}$ is 
computed. 
 \item By taking linear combinations, a basis $\tilde{P}$ of the vector space $\<P'\>$
is computed, such that each element of this basis has 
at most one monomial in $\partial B$ (and the other in $B$).
 \item The $C$-polynomials of the elements of $P$ with their support in
$B^{+}$ are computed and reduced by $\tilde{P}$.
 \item If non-zero polynomials with support in $B$ appear, the potential
basis $B$ and the polynomial set $P$ are updated. 
The update of $B$ is done by removing some parts of $B$.
The update of $P$ is done by combining the elements of $\tilde{P}$ and
the reduced $C$-polynomials, in order to get rewriting rules for the new set $B$.
\end{enumerate}
The details and the technical proof of termination of the algorithm are
available in \cite{mt:gnfpss-issac-05}.
%
% It may appear that the conjecture is wrong, either finding out one
% relation between the monomials of the vector space, then one has to
% remove some monomials from the basis of $B$. One may also come up the
% fact that there exists no relation between a monomial of the border of
% $B$ and the monomials of $B$, in that case one have to add this
% monomial to the basis of $B$.
%
%We refer to \cite{mt:gnfpss-issac-05} for an exposition of such an
%algorithm. 
We mention here that we have improved the algorithm
described in this paper since then: the ways the degree drops are
treated in \cite{mt:gnfpss-issac-05} is now done to avoid to the repeated
execution of similar reductions.

\section{Stability of the bases}\label{sec:stability}

This section is devoted to the study of the stability under numerical
perturbations, of the bases computed by the previous algorithm.

In many real-life problems, the system $\fb= (f_1,\ldots,f_s)$ to be solved is
given only with limited accuracy. However, most of the time, one also knows
that the structure of the solutions is invariant in a small neighborhood of
the system. Hence one of the feature that is often required to polynomial
solvers is to produce a representation of the quotient algebra that is stable
in a small neighborhood of the initial system. The structural numerical
stability of the basis is expected in order to have a smooth behavior
of the coefficients of the representation in the neighborhood of $\fb =
(f_1,\ldots,f_s)$.  By definition, a neighborhood of $\fb$ is an open set in
the space of vector of polynomials $(h_1, \ldots, h_s)$ such that
$\Lambda(h_{i})\preceq \Lambda(f_{i})$ ($i=1,\ldots,n$) and which contains $\fb$. 
For $\epsilon>0$, we define by $N_{\epsilon}(\fb)$ the set of systems $(h_1,
\ldots, h_s)$ such that $\Lambda(h_{i})\preceq \Lambda(f_{i})$ and the
coefficient vector of $h_{i}$ is at most at distance $\epsilon$ from the
coefficient vector of $f_{i}$ (for the $\infty$-norm).

% XXX Pourquoi est-ce specifique a cette fonction gamma ?
% XXX L'argument de la preuve Ca doit marcher pour n'importe quel gamma ?
% XXX La definition de max et gamme_epsilon ne me semble pas utiles?
%
%  \begin{defi}
%    Let $p$ be a polynomial, in this section we call  by {\em maximal
%    part of $p$}, denoted by $\Max(p)$ the homogeneous component of
%    $p$ of maximal degree (see also \cite{MoSa02}).
%  \end{defi}
%

\begin{assume} \label{def:gamma-support}
%  Let $p$ be a polynomial of $R$, $\Max(p)$ be the maximal part of
%  $p$.  Suppose also given a mesure on $\kk$.
Hereafter, $\gamma$ denotes a choice function refining a reducing grading
$\Lambda$, such that for all $p\in R$, $\gamma(p)$ depends only on the
support of $p$ and not on the numerical value of its monomial coefficients 
(e.g. Macaulay's choice function, grevlex choice function,\dots).
Let $\gamma_{\epsilon}$ be the choice function that for any $p\in R$, applies
the choice function $\gamma$ on the monomials of $p$, which coefficient 
norm is bigger than $\epsilon$.
%  and that chooses as leading monomial of a polynomial $p$, the monomial of
%  $\supp(p)$ that is the greatest monomial for the lexicographic order
%  of the monomials of $\Max(p)$ whose coefficient is at most at
%  distance $\epsilon$ from the maximum of the coefficients of the
%  monomials of $\Max(p)$.
\end{assume}

\begin{theorem}\label{thm:stab1}
  Let $\fb=(f_1,\ldots,f_s)$ be a zero dimensional polynomial system
such that in a neighborhood $U$ of $\fb$, all systems have the same number 
$D$ of complex solutions, counted with multiplicities. 
%  Let $\epsilon$ denote an upper bound of the amplitude of the uncertainity.
%  Let us suppose also that small
%  perturbations on the coefficients of the $f_i$ do not change the
%  geometric properties of the solutions of $(f_1=0,\ldots,f_s=0)$. 
% In particular, let us suppose that the support of the maximal part of
%   the $f_i$ is left unchanged by these perturbations. 
Then for all $\epsilon>0$ small enough, there exists $\nu>0$ such
that for any system $\fb' \in N_{\nu}(\fb)
\subset U$, the basis $B$ computed with $\gamma$ satisfying Assumption
\ref{def:gamma-support} for the
system $\fb$ is also a basis for the system $\fb'$.
%algorithm \ref{algonf} used with $\gamma_{\epsilon}$ computes the same basis $B$.
\end{theorem}

\begin{proof}
%   By hypothesis, all systems in the neighborhood of $\fb$ have the same number $D$
%   of complex solutions.
%   First of all, remark that if the geometric properties of the variety
%   defined by $(f_1,\ldots,f_s)$ are invariant with respect to small
%   perturbations, then the total number of solutions of the perturbed
%   systems is invariant with respect to small perturbations. Now we
%   will denote by $D$ the number of solutions of $I=(f_1,\ldots,f_s)$.
% 
  Let us consider the matrix $M$ whose rows correspond to the
  coefficient vectors of the monomial multiples $m\,f_i$, with
  $\deg(m)+\deg(f_i)\leq \kappa+d_{0}$ where $\kappa$ is the number of loops
  in Algorithm \cite{mt:gnfpss-issac-05} and $d_{0}$ the maximum degree of
the polynomials $f_{i}$. The columns are indexed by all the monomials of degree
  $\leq \kappa+d_{0}$. 

  We denote by $B$ the monomial set obtained as a basis of $\Ac=R/I$ by
  applying Algorithm \cite{mt:gnfpss-issac-05} to $f_1,\ldots,f_s$ with the
  choice function $\gamma$. By construction, the set of monomials indexing
  the columns of $M$ contains $\partial B$.

  Let $M^g$ be the same matrix as $M$ but constructed with the polynomials
  $f_{i}$ replaced by {\em generic } equations, i.e.  equations with
  indeterminate coefficients having the same support as $f_1,\ldots,f_s$.
  Let $N$ be the block of columns of $M$ indexed by monomials not in
  $B$ and let $N^{g}$  be the corresponding block in $M^{g}$.

  Since the operations of Algorithm \cite{mt:gnfpss-issac-05} consist in
  computing linear combinations of some monomial multiples of the polynomials
  $P$ (see Section \ref{sec:algo} and \cite{mt:gnfpss-issac-05})
  and thus of monomial multiples of $f_{i}$ of degree $\leq
  \kappa+d_{0}$, the complete computation can be reinterpreted as an
  optimized triangulation procedure of the block $N$.

  The block $N^{g}$ specialized at $\fb$ is invertible, because any monomial
  not in $B$ can be reduced by the computed border basis of $\fb$ to an
  element in $\<B\>$. This implies that $N^{g}$ specialized at $\fb'$ is also
  invertible, for $\fb' \in N_{\nu}(\fb)$ with $\nu>0$ sufficiently
  small. Therefore any monomial of $\partial B$ can be reduced modulo the
  ideal $(\fb')$ to an element in $\<B\>$. Consequently, $B$ (which contains
  $1$) is a generating set of the quotient algebra $\Ac'=R/(\fb')$ for any
  $\fb' \in N_{\nu}(\fb)$.

  As the number of solutions $D=|B|$ (counted with multiplicity) is left
  unchanged by small perturbations in the neighborhood $N_{\nu}(\fb)$ of
  $\fb$, the monomial set $B$, that has exactly cardinality $D$, is also a basis
  of the quotient algebra $\Ac'=R/(\fb')$ for the perturbed system
  $\fb'$. 
\end{proof}

Consider now a slight modification of Algorithm 
\cite{mt:gnfpss-issac-05}: the
coefficients whose norm is less than $\epsilon$ are simply ignored in all the steps of Algorithm
\cite{mt:gnfpss-issac-05}. This means that they will not be taken into
account for choosing a {\em leading} monomial, or deciding if a
polynomial is nonzero. This small variant will be denoted as
the $\epsilon$-algorithm in the next theorem.  Remark here that this
behavior is quite classical in fact, it is more or less what is usually
done when neglecting small coefficients in a numerical algorithm.
Remark also that the following theorem performs rigorously, and that
the computed result is {\em not} an approximation of the true quotient
algebra!

Then the following holds:
\begin{theorem}\label{thm:stab2}
  Let $\fb=(f_1,\ldots,f_s)$ be a zero dimensional polynomial system such
  that in a neighborhood $U$ of $\fb$, all systems have the same number $D$
  of complex solutions, counted with multiplicities.  Then for all
  $\epsilon>0$ small enough, there exists $\nu>0$ such that for any system
  $\fb' \in N_{\nu}(\fb) \subset U$, the basis $B$ computed with $\gamma$
  satisfying Assumption
\ref{def:gamma-support} for the system $\fb$ is also the basis
  obtained with $\gamma_\epsilon$ and the $\epsilon$-algorithm for the system
  $\fb'$.  \end{theorem}

\begin{proof}
  By Theorem \ref{thm:stab1}, $B$ is also a basis of the quotient
  algebra $R/(\fb')$.  This basis $B$ is obtained by applying
  Algorithm \cite{mt:gnfpss-issac-05} to $\fb$. The result of this
  algorithm does not change for the system $\fb$ if we replace the
  choice function $\gamma$ by $\gamma_{\epsilon}$ for $\epsilon>0$
  small enough (eg. smaller than the minimum of the norm of the
  coefficients of the polynomials on which $\gamma$ is applied).

  Let us show by induction on the loop index $k$ of the algorithm that the
  steps and the polynomials computed by the $\epsilon$-algorithm with
  $\gamma_\epsilon$ are the same as for the direct algorithm, up to the terms
  of norm smaller that $\epsilon$. 

  It is true obviously for the first step $k=1$.  Let us suppose now that
  steps $1,\ldots, k'$ of Algorithm \cite{mt:gnfpss-issac-05} ran on $\fb$
  and its $\epsilon$ variant ran on $\fb'$ are structurally the same and let
  us show that step $k'+1$ is also structurally the same for the two
  computations. The coefficients of all these constructed polynomials are
  rational functions of the coefficients of $\fb'$, which are well defined in
  a neighborhood of $\fb$. If $\fb'$ is close enough to $\fb$, the monomials
  which are in the support of the constructed polynomials for $\fb'$ but not in
  the support of the constructed polynomials for $\fb$, have coefficients
  of norm smaller then $\epsilon$.

  If, by hypothesis, the first $k'$ steps are structurally identical, the
  same polynomials, up to terms of norm smaller than $\epsilon$, appear when
  selecting the polynomials in $P^{+}$, and the same $C$-polynomials are
  constructed (see Section \ref{sec:algo} and \cite{mt:gnfpss-issac-05}). By
  continuity of the coefficients of the constructed polynomials, the {\em
  same} pivots (of norm bigger than $\epsilon$) are used to construct the new
  elements in $\tilde{P}$. Similarly, in a neighborhood of $\fb$, up to terms
  of small norm, the $C$-polynomials not reducing to zero are the same for
  the two computations. If choices of leading monomials are to be performed,
  then, by Definition \ref{def:gamma-support}, $\gamma_{\epsilon}$ will
  select the same monomials for $\fb'$ and $\fb$. Finally at the end of step
  $k'+1$ the same computations are performed, up to terms of small norm and
  the coefficients of the new constructed polynomials are rational functions
  of the coefficients of $\fb'$, which are well defined in a neighborhood of
  $\fb$.

  This ends the induction showing that the two computations are structurally
  identical for $\fb'\in N_{\nu}(\fb)$ with $\nu>0$ small enough. Hence
  $B$ will also be found as a basis of $R/(\fb')$ with this
  $\epsilon$-algorithm.
\end{proof}

The rewriting rules obtained from the $\epsilon$-algorithm are close to the
exact rewriting rules of the system $\fb'$. Their numerical quality can be
improved by iterative refinements such as Newton-like iterations using the
commutation relations. Such approach has been investigated in
\cite{Stetter04}.

\begin{remark}[Numerical certification]
Theorem \ref{thm:stab2} shows the continuity of the normal form computation with
respect to the coefficient of the input system. It states that there exists a
region of {\em stability} for the computed quotient algebra representation,
but it is an open problem to compute apriori the value $\epsilon$ and
$\nu$ for a given polynomial system in order to control the size of the
allowed perturbations. This problem is the subject of further work.
\end{remark}
\begin{remark}[Flatness]
Theorem \ref{thm:stab2} also shows that if we consider a rationally
parametrized family of systems $\fb_{t} \in N_{\nu}(\fb)$ for all $t\in
[0,1]$, such that
$\fb_{0}=\fb$ and $F_{0}$ is the border basis for $B$, then the set $B$ is
also a basis of $\Ac_{t}=R/(\fb_{t})$.  Moreover, the border basis $F_{t}$ of
$\fb_{t}$ for $B$ is of the form 
$F_{t}= (\omega - \rho_{\omega,t})_{\omega\in \partial B}$, 
where $\rho_{\omega,t}\in \<B\>$ is a continuous (rational) function of $t$
on $[0,1]$ such that $F_{0}=(\omega - \rho_{0,\omega})_{\omega\in \partial
B}$.  This also implies that the $C$-relations \eqref{syz:next-door} and
\eqref{syz:across-street} generating the syzygies of $F_{t}$ are continuous (rational)
functions of $t\in [0,1]$, which coincide with the $C$-polynomials relations of $F_{0}$
at $t=0$. Consequently any syzygy of $F_{0}$ which is a combination of the
$C$-relations can be deformed continuously into a syzygy of $F_{t}$ (for
$t\in [0,1]$).  In other words, the systems $\fb'$ in the neighborhood
$N_{\nu}(\fb)$ are {\em flat deformations} of the system $\fb$
\cite{Eis94}.
\end{remark}

\section{Experimentation}
The algorithm described in the previous section is implemented in the library
\textsc{Synaps}\footnote{\texttt{http://www-sop.inria.fr/galaad/synaps/}}.
It corresponds to about 50 000 lines of C++-code. It involves a direct sparse
matrix solver. The numerical approximation of the roots are obtained by
eigenvalues computation, using the library {\sc lapack} (the routine
\texttt{zgees}) and the strategy described in \cite{CGT97}.
The computations are performed on an AMD-Athlon 2400+ with $256MB$ of
main memory. We show the results obtained with our implementation in
the case where the grading that we use for $\kx$ is the usual one.
In the sequel, $drvl$ will refer to the choice function associated
to the Degree Reverse Lexicographical order, $dlex$ to the degree
lexicographical order, %$random$ to the choice function that returns
%randomly any of the monomials of maximum degree of the polynomial
%given as its input,
$Mac$ to Macaulay's choice function (see Example \ref{macfun}), $minsz$ to the
choice function over the rational that minimizes the memory needed in
the reduction loop (this choice only minimizes a local step and does
not insure local minimality of the global required memory), and $mix$ to the
choice function that returns randomly
either the result of $minsz$, or of $drvl$ applied to
its input. To analyse the quality of approximation, we mesure 
the maximal norm at the computed roots of the initial polynomials $f_{i}$ and
denote it hereafter by {\it mnacr}.

\subsection{Generic equations}
The method we propose here is an extension of the Gr\"obner bases
computations. As such it can compute Gr\"obner bases. The
implementation we have is not as optimized as the Gr\"obner bases ones
that are being worked on for decades. An important
work, mostly on linear algebra, remains to be done on our program.
However we want to show that the method we propose here is competitive,
and that it does not lose the good practical efficiency of Gr\"obner
bases computations \cite{Faug99}. As the arithmetic used in our
programs for doing exact computation is the rational arithmetic of
{\sc gmp}, which is much slower than integer computations used in the
other software we will restrict our-self to the use of modular
arithmetic.  The family of examples we have chosen is the
$Katsura(n)$\footnote{\texttt{http://www-sop.inria.fr/galaad/data/}}
equations. These equations are projective complete intersection with
no zero at infinity. Using the Macaulay choice function, we know
apriori that Macaulay's basis will be a monomial basis of the quotient
algebra, so we know, apriori, what monomials will be leading monomials
for the whole computation; so in this case we can guarantee that no
test to $0$ returns erroneous result even using floating point
arithmetic.  We compare first our program to one of the best
implementations available, Magma's implementation of $F_4$ algorithm
\cite{Faug99}.  {\small
  \begin{center}\small
%\vskip -\the\baselineskip 
\begin{tabular}[]{|c|c|c|c|c|c|c|}\hline
n&\multicolumn{2}{c|}{Synaps
  mac}&\multicolumn{2}{c|}{Synaps drvl}&\multicolumn{2}{c|}{Magma drvl}\\
\hline
    7 &  0.19s & 3M & 0.22s & 3M & 0.05s & 3M\\
    \hline
    9 & 6.17s & 5M & 8.44s & 5M & 1.670s & 7M\\
    \hline
    10 & 32.39s & 14M & 56.84s & 13M & 13.50s & 23M\\
    \hline
    11 & 252.05s & 50M & 387.97s  & 45M & 96.76s & 70M\\
    \hline
    12 & 1935.25s & 191M & 3072.08s  & 157M & 1560.76s & 240M\\
    \hline

\end{tabular}
%\vskip -\the\baselineskip 
\end{center} 
}

Let us mention that Gb, one of the reference implementation
of Buchberger's algorithm, spends {\small $659s$} on {\small $Katsura(10)$}.

Numerically we observe that choosing the {\it mac} function also
results in a better conditioning of the computations. More precisely
on {\em \small Katsura(6)}, and using a threshold of $10^{-10}$ we have:
{\small
\begin{center} \small
%\vskip -\the\baselineskip 
\begin{tabular}{|c|c|c|c|c|c|c|c|c|c|}\hline
$\gamma$    & drvl     & dlex      &  mac & drvl  & dlex      &  mac & drvl
& dlex      &  mac\\  \hline
\# bits  & 128      & 128       & 128  &  80      &  80       &  80  &  64      &  64       &  64 \\\hline
 time           & 1.98s    & 2.62s     & 1.64s& 1.35s    &  3.98s    &  0.95s
& $-$ & $-$ & 0.9s \\ \hline
mnar   
                &$10^{-28}$&$10^{-24}$ & $10^{-30}$
&$10^{-20}$&$10^{-15}$&$10^{-19}$& $-$& $-$& $10^{-11}$\\
\hline
\end{tabular}
\end{center}
}
For the $64$ bits computation the results computed for the {\it drvl} and
{\it dlex} orders are erroneous due to roundoff errors.
The time given is the time spent in the computation of the
multiplication matrices. Afterward, we used either {\sc lapack} to
perform the eigenvector computations or Maple when we needed extended
precision. Because of the different nature of these tools, we do not
report on the solving part timing.
Finally we show here the amount of memory needed to perform the computations
over $\qq$, using {\sc gmp}  {\tt mpq}. 

{\small
\begin{center}
%\vskip -\the\baselineskip 
\begin{tabular}{|c|c|c|c|c|}\hline
 &mac & minsz & drvl & mix \\
\hline
time &  4.22s & 30.21s & 6.54s & 7.83s \\
    \hline
 size & 4.2M & 6.1M & 4.4M & 4.9M\\
\hline
\end{tabular} 
\end{center}}
On these experiments, we observe that the local strategy {\it minsz} which tends to
minimize locally the size of the coefficients in the linear algebra
operations, does not yield globally the optimal output size. In this example,
the time and the memory size seem to be correlated.
\subsection{Parallel robot}\label{ROB}
Let us consider the famous direct kinematic problem of the parallel
robot\footnote{\texttt{http://www-sop.inria.fr/galaad/data/}}
%\cite{Laz92},
\cite{MB93issac}. 
First we use floating point numbers to check the numerical requirements of
the computations for different orders. For testing a number to be $0$, we will
use a leveling (here $10^{-8}$ is enough) and we will check afterward that the
choices performed are the same as those done using modular arithmetic. This
is equivalent to the use of an hybrid arithmetic \cite{GBBIB433}.
{\small
\begin{center}
%\vskip -10pt
\begin{tabular}{|c|c|c|c|}\hline
$\gamma$ & \# bits &  time & mnacr \\
%\hline
%random& 128 & 16.5s& $0.5*10^{-22}$\\
\hline
drvl& 128 & 2.07s& $0.3*10^{-24}$\\
\hline
dlex& 128 & 4.27s& $0.3*10^{-23}$\\
\hline
mac& 128 &2.22s & $0.1*10^{-24}$\\%2.22
%\hline
%dlex& 250 & 11.16s& $0.42*10^{-63}$\\
%\hline 
%dlex & 250 & 13.8s & $0.135*10^{-60}$\\
%\hline
%mac & 250 &  11.62s& $0.46*10^{-63}$\\
\hline
\end{tabular}
%\vskip-9pt
\end{center}
}
Here we see that choosing the right choice function can increase (but not so
much in this case) the numerical accuracy of the roots. Hereafter we
use the parametrization of \cite{Laz92} for solving, it involves
more variables, gives better timings but less correct digits on the
final result.

%\vspace{-12pt}
{\small
\begin{center}
\begin{tabular}{|c|c|c|}\hline
\# bits &  time & mnacr\\
\hline
250 &  1.32s& $10^{-63}$ \\
\hline
500 &  2.23s & $10^{-140}$ \\
\hline
\end{tabular}\end{center}}
%\vspace{-10pt}

Finally we performed tests using rational arithmetic. 

%\vspace{-10pt}
{\small
\begin{center}
%\vskip -\the\baselineskip 
\begin{tabular}{|c|c|c|c|c|}\hline
$\gamma$ &mac & minsz & drvl & mix\\
\hline
 time &  315s & 229.08s & 201.65s & 257.50 \\
    \hline
 size & 17M & 14M & 16M & 13M\\
     \hline
\end{tabular}
%\vskip -\the\baselineskip 
\end{center}}

In fact, it is not so surprising to see that the choice function $\gamma$ has
a big impact in terms of the computational time and of the memory
required.
However in this problem, the time and the memory size do not seem
to be correlated as in the previous case. 

We also mention here that over-constraining the system can result in a
dramatic decrease of the computation time. Indeed expressing more constraints
than necessary can simplify computations
significantly (see \cite{mt:gnfpss-issac-05}).

\noindent{}{\bf Acknowledgments:} We thank A. Quadrat for interesting discussions
on differential algebra, prolongation and involutivity. 

\negskip

%\bibliographystyle{plain}
%\bibliography{tcs}
%\bibliography{algebra,geometry,mourrain,trebuchet,robotics,linear,info,numeric}

\end{document}